\documentclass{llncs}

\usepackage{graphicx}
\usepackage{color}
\usepackage{algorithm}
\usepackage[noend]{algorithmic}

\usepackage{amsmath}
\usepackage{amssymb}
\usepackage{latexsym}
\usepackage{array}
\usepackage{stmaryrd}

\usepackage{paralist}
\usepackage{graphicx} 
\usepackage{xspace}
\usepackage{url}
\usepackage{tabulary,booktabs,multirow}
\usepackage{algorithm, algorithmic}
\usepackage{tikz}
\usepackage{footnote}
\usepackage{subfig}
\usepackage[shortlabels]{enumitem}
\setlist{nolistsep}

\usepackage{tikz}
\usetikzlibrary{positioning,arrows,automata,shapes,decorations.pathreplacing,patterns}

\usepackage{ulem}
\newcommand{\res}{\mathsf{res}}

\newcommand{\ur}{\mathit{UR}}
\newcommand{\etal}{\emph{et al.}\xspace}
\newcommand{\ie}{\emph{i.e.}\xspace}

\newcommand{\eg}{\emph{e.g.}\xspace}

\newcommand{\SAT}{\textsc{SAT}\xspace}
\newcommand{\TDM}{\textsc{$3$-Dimensional Matching}\xspace}

\newcommand{\RCPlong}{\textsc{Resiliency Checking Problem}\xspace}
\newcommand{\RCPshort}{\textsc{RCP}\xspace}
\newcommand{\RCP}[1]{\ensuremath{\RCPshort \langle #1 \rangle}}

\newcommand{\ILPF}{\textsc{ILPF}\xspace}

\newcommand{\Res}{R}
\newcommand{\p}{p}
\renewcommand{\P}{P}
\newcommand{\Users}{U}

\newcommand{\U}{\mathcal{U}}

\renewcommand{\S}{\mathcal{S}}
\newcommand{\C}{\mathcal{C}}

\usepackage{makeidx}  
\begin{document}
\frontmatter          
%
%
%
\mainmatter              
\title{A Multivariate Approach for Checking Resiliency in Access Control\protect\footnote{This research was partially supported by EPSRC grant EP/K005162/1. Gutin's research was also supported by Royal Society Wolfson Research Merit Award.}}
%
%
\author{Jason Crampton \and Gregory Gutin \and R\'emi Watrigant}
%
%
%
\institute{Royal Holloway University of London}


\date{}
\maketitle%
\begin{abstract}
In recent years, several combinatorial problems were introduced in the area of access control. 
Typically, such problems deal with an authorization policy, seen as a relation $\ur \subseteq \Users \times \Res$, where $(u, r) \in \ur$ means that user $u$ is authorized to access resource $r$.
Li, Tripunitara and Wang (2009) introduced the \RCPlong (\RCPshort), in which we are given an authorization policy, a subset of resources $P \subseteq \Res$, as well as integers $s \ge 0$, $d \ge 1$ and $t \geq 1$. 
It asks whether upon removal of any set of at most $s$ users, there still exist $d$ pairwise disjoint sets of at most $t$ users such that each set has collectively access to all resources in $P$. 
This problem possesses several parameters which appear to take small values in practice. We thus analyze the parameterized complexity of \RCPshort with respect to these parameters, by considering all possible combinations of $|P|, s, d, t$. In all but one case, we are able to settle whether the problem is in FPT, XP, W[2]-hard, para-NP-hard or para-coNP-hard. We also consider the restricted case where $s=0$ for which we determine the complexity for all possible combinations of the parameters.
\end{abstract}

\section{Introduction}\label{sec:intro}
\subsection{Context and definition of the problem}

Access control is a fundamental aspect of the security of any multi-user computing system. Typically, it is based on the idea of specifying and enforcing an authorization policy, identifying which interactions between a set of users $\Users$ and a set of resources $\Res$ are to be allowed by the system~\cite{SaCoFeYo96}. More formally, an authorization policy is defined as a relation $\ur \subseteq \Users \times \Res$, where $(u, r) \in \ur$ means that user $u$ is authorized to access resource $r$.  
Quite recently, we have seen the introduction of resiliency policies, whose satisfaction indicates that a system will continue to function as intended in the absence of some number of authorized users \cite{LiWaTr09,WaLi10}. 
Li, Tripunitara and Wang's seminal work \cite{LiWaTr09} introduces a number of problems associated with the satisfaction of a resiliency policy.  One of their motivating examples concerns the emergency response to a natural disaster, where teams of users must perform the same critical operation(s) at multiple (distinct) geographical locations.  Thus the members of each team must be authorized collectively to perform the operation(s).  In addition, we may wish to impose an upper bound on the size of the teams because, for example, of constraints on transportation.

For a user $u \in \Users$ and a set of users $V \subseteq \Users$, we define $N_{\ur}(u) = \{r \in \Res : (u, r) \in \ur\}$ the \emph{neighborhood} of $u$ and, by extension, $N_{\ur}(V) = \bigcup_{u \in V} N_{\ur}(u)$ the \emph{neighborhood} of $V$, omitting the subscript $\ur$ if the authorization policy is clear from the context.
Given an authorization policy $\ur \subseteq \Users \times \Res$, an instance of the \RCPlong (\RCPshort) is defined by a resiliency policy $\res(P, s, d, t)$, where $P \subseteq \Res$, $s \ge 0$, $d \ge 1$ and $t \ge 1$. We say that $\ur$ \emph{satisfies} $\res(P, s, d, t)$ if and only if for every subset $S \subseteq \Users$ of at most $s$ users, there exist $d$ pairwise disjoint subsets of users $V_1, \dots, V_d$ such that for all $i \in \{1, \dots, d\}$:
%
\begin{align}
&V_i \cap S = \emptyset, \label{def:cond:intersect} \\
&|V_i| \le t, \label{def:cond:size} \\
&N(V_i) \supseteq P. \label{def:cond:neighb}
\end{align}

We are now ready to define the main problem we study in this paper:\\

\fbox{
\begin{minipage}[c]{11cm}
\textbf{\RCPlong (\RCPshort)}\\
\underline{Input:} $\ur \subseteq \Users \times \Res$, $P \subseteq \Res$, $s \ge 0$, $d \ge 1$, $t \ge 1$.\\
\underline{Question:} Does $\ur$ satisfy $\res(P, s, d, t)$ ?
\end{minipage}}~\\ 

Furthermore, we will adopt the bracket notation \RCP{} used by Li \etal~\cite{LiWaTr09} to denote some restrictions of the problem, in which one or more parameters (among $s$, $d$ and $t$) are fixed. In particular, we will consider the cases where $s$ and $d$ are respectively set to $0$ and/or $1$ (or other fixed positive values), while $t$ might be set to $\infty$, meaning that there is no constraint on the size of the sets (which is actually equivalent to $t=|P|$, implying that we may assume in the remainder that $t \le |P|$). For instance, \RCP{s=0} denotes the variant in which $s$ is fixed to $0$, \ie we ask for the satisfaction of $\res(P, 0, d, t)$. In the remainder of the paper, we set $p=|P|$.

Given an instance of \RCP{}, we say that a set of $d$ pairwise disjoint subsets of users $V = \{V_1, \dots, V_d\}$ satisfying conditions~\eqref{def:cond:size} and~\eqref{def:cond:neighb} is a \emph{set of teams}. For such a set of teams, we define $\U(V) = \bigcup_{i=1}^d V_i$.
Given $\Users' \subseteq \Users$, the \emph{restriction} of $\ur$ to $\Users'$ is defined by $\ur|_{\Users'} = \ur \cap (\Users' \times \Res)$. 
Finally, a set of users $S \subseteq \Users$ is called a \emph{blocker set} if for every set of teams $V=\{V_1, \dots, V_d\}$, we have $\U(V) \cap S \neq \emptyset$. 
 Equivalently, observe that $S$ is a blocker set if and only if $\ur|_{\Users \setminus S}$ does not satisfy $\res(P, 0, d, t)$.
 Throughout the paper, we write $[d]$ to denote $\{1, \dots, d\}$ for any integer $d \ge 1$, and we will often make use of the $O^*(.)$ notation, which omits polynomial factors and terms.
 
\subsection{Parameters}

An instance of \RCP{} contains several parameters (namely $s$, $d$ and $t$) which may be used for the complexity analysis of the problem. 
An interesting point of the work of Li \etal~\cite{LiWaTr09} is that the number of users in an organization will typically be large in comparison to the other parameters ($s$, $d$, $t$, and even $p$) in practice. 
In their experiments, the maximum values used are $n=100$, $p=10$ and $d=7$ (they only run experiments on the variant where $t=\infty$, but, as we observed previously, we may set $t=p$).
 With this in mind, we exploit the theory of fixed-parameter tractability in order to settle the parameterized complexity of the problem.

Given an instance $x$ (of size $|x|$) of a decision problem, with some parameter\footnote{Note that one can aggregate several parameters $p_1, \dots, p_m$ by defining $k=p_1+\dots+p_m$, in which case we will say the parameter is $(p_1, \dots, p_m)$.} $k$, we are interested in algorithms deciding whether $x$ is positive or negative in polynomial time when $k$ is bounded above by a constant. 
More precisely, if such an algorithm has running time $O(f(k)|x|^{O(1)})$ for some computable function $f$, then we will say that this algorithm is fixed-parameter tractable (FPT), while if its running time is $O(|x|^{f(k)})$ for some computable function $f$, we will say that this algorithm is XP (an FPT algorithm is thus an XP algorithm). 
By extension, FPT (resp. XP) gathers all problems for which an FPT (resp. XP) algorithm exists.
Proving the NP-hardness of a problem in the case where a parameter $k$ is bounded above by a constant immediately forbids the existence of any XP (and thus FPT) algorithm unless P $=$ NP. 
In this case, we will say that this parameterized problem is \mbox{para-NP-hard}. 
A similar definition can be given using coNP-hard and coNP instead of NP-hard and NP, respectively, leading to the para-coNP-hard complexity class (and thus, if a problem is shown to be para-coNP-hard, then it does not belong to XP unless P $=$ coNP). In the following, para-(co)NP-hard denotes the union of para-NP-hard and para-coNP-hard.
Finally, the W[$i$]-hierarchy of parameterized problems is typically used to rule out the existence of an FPT algorithm, under the widely believed conjecture that FPT $\neq$ W[$1$].
For more details about fixed-parameter tractability, we refer the reader to the recent monographs \cite{CyFoKoLoMaPiPiSa15,DoFe13}.

\subsection{Related work}

As one might expect, the \RCP{} problem is strongly related to some known combinatorial problems. 
Indeed, one can observe that \RCP{s=0, d=1} is equivalent to the \textsc{Set Cover} problem, while \RCP{s=0, t=\infty} can be reduced in a straightforward way from the \textsc{Domatic Partition} problem (in the \textsc{Domatic Partition} problem, one asks whether a given graph admits $k$ pairwise disjoint dominating sets). 
Li \etal~\cite{LiWaTr09} obtained several (mainly negative) results for \RCP{} in some restricted cases which can be summarized by the following theorem.

\begin{theorem}[\cite{LiWaTr09}] \label{thm:lietal} We have the following:
  \begin{itemize}
  	\item \RCP{}, \RCP{d=1} and \RCP{t=\infty} are NP-hard and are in\footnote{coNP$^{\text{NP}}$ is the set of problems whose complement can be solved by a non-deterministic Turing machine having access to an oracle to a problem in NP.} coNP$^{\text{NP}}$;
  	\item \RCP{s=0, d=1}, \RCP{s=0, t=\infty} are NP-hard;
  	\item \RCP{d=1, t=\infty} can be solved in linear time.
  \end{itemize}
\end{theorem}

In addition, they developed and implemented an algorithm for \RCP{} which consists of%
\begin{inparaenum}[(i)] 
	\item enumerating all subsets of at most $s$ users, and \label{firststep}
	\item for each such subset $S$, determining the satisfaction of $\res(P, 0, d, t)$ for $\ur|_{\Users \setminus S}$.\label{secondstep}
\end{inparaenum}
Step (\ref{secondstep}) is achieved by a \SAT formulation of the problem and the use of an off-the-shelf \SAT solver, while they develop a pruning strategy in order to avoid the entire enumeration of all subsets of users of size at most $s$, resulting in an efficient speed-up of \mbox{step (\ref{firststep})}. Quite surprisingly, they observe that the bottleneck of their algorithm lies in the second step, where an instance of \RCP{s=0} has to be solved. This motivated us to focus on the parameterized complexity of \RCP{s=0} separately.

\subsection{Contribution and organization of the paper}\label{sec:contribution}

Our goal in this paper is thus to determine the parameterized complexity of \RCP{} and \RCP{s=0} with respect to parameters $p, s, d, t$, by considering every possible combination of them. 
In each case, we aim at determining whether the problem is%
\begin{inparaenum}[(i)]
	\item in FPT, \label{case:fpt}
	\item in XP but W[$i$]-hard for some $i \ge 1$, or \label{case:xp}
	\item para-(co)NP-hard. \label{case:parahard}
\end{inparaenum}


\begin{figure}[!b]
  \begin{center}
    
\begin{tikzpicture}%
  [fpt/.style={rectangle,draw,minimum width=1cm},%
   xp/.style={rectangle,draw,pattern=north west lines,pattern color=black!60,minimum width=1cm},%
   paranphard/.style={rectangle,draw,fill=black!30,minimum width=1cm},%
   unknown/.style={circle,draw}]
  \node[fpt] (pdt) at (0,0) {$p,d,t$};
  \node[fpt] at (-1.5,-2) (pd) {$p,d$};
  \node[fpt] at (-1.5,-4) (p) {$p$};
  \node[fpt] (pt) at (0,-2) {$p,t$};
  \node[xp] (dt) at (1.5,-2) {$d,t$};
  \node[paranphard] (t) at (1.5,-4) {$t$};
  \node[paranphard] (d) at (0,-4) {$d$};
  \draw[->] (pdt) -- (pd);
  \draw[->] (pdt) -- (pt);
  \draw[->] (pdt) -- (dt);
  \draw[->] (pd) -- (p);
  \draw[->] (pd) -- (d);
  \draw[->] (dt) -- (d);
  \draw[->] (dt) -- (t);
  \draw[->] (pt) -- (p);
  \draw[->] (pt) -- (t);
  \node[fpt, minimum width=12pt,label=right:FPT] (legend-fpt) at (-2,-6) {};
  \node[xp, minimum width=12pt,label=right:{W[2]-hard but XP}] (legend-xp) at (0,-6) {};
  \node[paranphard, minimum width=12pt,label=right:para-(co)NP-hard] (legend-para) at (4,-6) {};
  \node[unknown, minimum width=12pt,label=right:open] (legend-unknown) at (8,-6) {};

  \node[fpt] (psdt) at (5.75,1) {$p,s,d,t$};

  \node[fpt] at (3,-1) (pst) {$p,s,t$};
  \node[fpt] at (4.75,-1) (psd) {$p,s,d$};
  \node[fpt] at (6.5,-1) (pdt) {$p,d,t$};
  \node[xp] at (8.25,-1) (sdt) {$s,d,t$};

  \node[unknown] (pt) at (2.5,-3) (pt) {$p,t$};
  \node[fpt] at (3.8,-3) (ps) {$p,s$};
  \node[fpt] at (5.1,-3) (pd) {$p,d$};
  \node[paranphard] (sd) at (6.4,-3) {$s,d$};
  \node[paranphard] (st) at (7.7,-3) {$s,t$};
  \node[paranphard] (dt) at (9,-3) {$d,t$};

  \node[unknown] (p) at (3,-5) (p) {$p$};
  \node[paranphard] (s) at (4.75,-5) (s) {$s$};
  \node[paranphard] (t) at (6.5,-5) (t) {$t$};
  \node[paranphard] (d) at (8.25,-5) (d) {$d$};
  \draw[->] (psdt) -- (psd);
  \draw[->] (psdt) -- (pst);
  \draw[->] (psdt) -- (pdt);
  \draw[->] (psdt) -- (sdt);
  \draw[->] (pdt) -- (pd);
  \draw[->] (pdt) -- (pt);
  \draw[->] (pdt) -- (dt);
  \draw[->] (pst) -- (ps);
  \draw[->] (pst) -- (pt);
  \draw[->] (pst) -- (st);
  \draw[->] (psd) -- (ps);
  \draw[->] (psd) -- (pd);
  \draw[->] (psd) -- (sd);
  \draw[->] (sdt) -- (sd);
  \draw[->] (sdt) -- (st);
  \draw[->] (sdt) -- (dt);
  \draw[->] (pd) -- (p);
  \draw[->] (pd) -- (d);
  \draw[->] (dt) -- (d);
  \draw[->] (dt) -- (t);
  \draw[->] (pt) -- (p);
  \draw[->] (pt) -- (t);
  \draw[->] (ps) -- (p);
  \draw[->] (ps) -- (s);
  \draw[->] (pd) -- (p);
  \draw[->] (pd) -- (d);
  \draw[->] (sd) -- (s);
  \draw[->] (sd) -- (d);
  \draw[->] (st) -- (s);
  \draw[->] (st) -- (t);
  \draw[->] (dt) -- (d);
  \draw[->] (dt) -- (t);
  
 \end{tikzpicture}
 
    \caption{Schemas of the complexity of \RCP{s=0} (left) and \RCP{} (right) after the results obtained in this paper (see the end of this section for the difference between old and new results).}
    \label{fig:lattices}
  \end{center}
\end{figure}
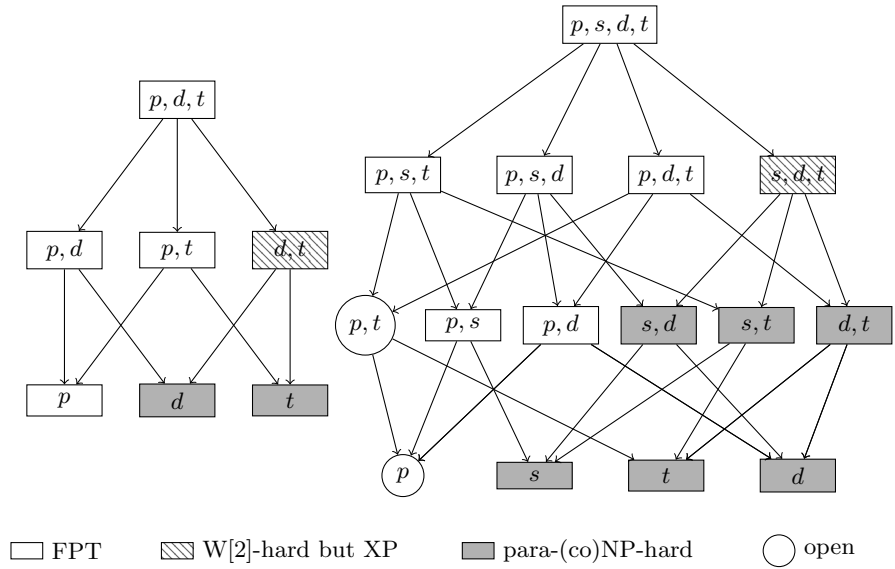

Figure~\ref{fig:lattices} summarizes the (already known and) obtained results for \RCP{} and \RCP{s=0} with respect to all possible combinations of the parameters specified previously.
An arrow $A \longrightarrow B$ means that $A$ is a larger parameter than $B$, in the sense that an FPT algorithm parameterized by $B$ implies an FPT algorithm parameterized by $A$, and, conversely, any negative result parameterized by $A$ implies the same negative result parameterized by $B$.
Since (under classical complexity assumptions) a decision problem is either in one of the previous cases (\ref{case:fpt}), (\ref{case:xp}) or (\ref{case:parahard}), one can observe that the parameterized complexity of \RCP{s=0} is now completely determined with respect to all possible combinations of parameters $p$, $d$ and $t$. 
Concerning the more general case \RCP{}, only the parameterization by $p$ only remains unknown (recall that as we mentioned earlier, we may assume in any instance that $t \le p$, implying that adding $t$ in the parameter list is of no importance concerning the membership in these complexity classes, both for positive or negative results).

The next section gathers all our results for the general case \RCP{}, namely:
\begin{itemize}
	\item membership in XP parameterized by $(s, d, t)$ (Theorem~\ref{thm:xp-ps-sdt}),
	\item membership in FPT parameterized by $(p, d)$ or $(p, s)$ (Theorem~\ref{thm:rcpdp}),
	\item para-coNP-hardness parameterized by $(d, t)$ (Theorem~\ref{thm:paranphard-hittingset}),
	\item para-NP-hardness parameterized by $(s, t)$ (Theorem~\ref{thm:paranphard-matching}).
\end{itemize} 
Note that the para-NP-hardness for $(s, d)$ was already known (Theorem~\ref{thm:lietal}), as well as the W[$2$]-hardness for $(s, d, t)$ (see explanation in Section~\ref{sec:negres}).

Section~\ref{sec:sZero} gathers all our results for the restricted case \RCP{s=0}, namely:
\begin{itemize}
	\item an FPT algorithm parameterized by $(d, p)$ with an optimal running time (under ETH) when $d$ is fixed (Theorems~\ref{thm:DP} and \ref{thm:lowerbounds}),
	\item membership in FPT parameterized by $p$ only (Theorem~\ref{thm:rcp-parameterized-by-p-only}).
\end{itemize}
Note that the W[$2$]-hardness for $(d, t)$ is inherited from \RCP{}, while the XP membership results from a brute-force enumeration of all subsets of users of size $dt$.
We also investigate in this section the question of data (user) reductions and present positive and negative kernelization results depending on the considered variant: \RCP{s=0} or \RCP{s=0, t=\infty} (Theorem~\ref{thm:kernel}).
We finally conclude the paper in Section~\ref{sec:conclusion}.


\section{The general case}\label{sec:generalcase}

\subsection{Positive results}

First, observe that there exists a simple XP algorithm for \RCP{} parameterized by $(s, d, t)$. Indeed, recall that the problem actually aims to check whether there is a set $S \subseteq \Users$ of size at most $s$ such that for any set of teams $V=\{V_1, \dots, V_d\}$ we have $S \cap \U(V) \neq \emptyset$, and note that finding a set of teams is exactly the \RCP{s=0} problem, which is in XP parameterized by $(d, t)$, as said in Section~\ref{sec:contribution}. Hence, since $|\U(V)| \le dt$, by finding iteratively a set of teams and branching on each element to be removed from it (and included in the future blocker set), one can determine whether there exists a blocker set of size at most $s$ in XP time parameterized by $(s, d, t)$:

\begin{theorem} \label{thm:xp-ps-sdt}
\RCP{} is in XP when parameterized by $(s, d, t)$
\end{theorem}

Despite its simplicity, this result is actually somehow tight. First, as we will see later (Section~\ref{sec:negres}), \RCP{} is W[$2$]-hard with this parameterization. In addition, considering a strict subset of $\{s, d, t\}$ as a parameter makes the problem para-(co)NP-hard (Theorems~\ref{thm:lietal}, \ref{thm:paranphard-hittingset} and \ref{thm:paranphard-matching}).
A way of going further is to ``replace'' $t$ by $p$ (since we may assume $t \le p$). With this modification, we show in the next result how to get rid of the parameter $s$ or $d$ by designing an FPT algorithm parameterized by $(p, d)$ or $(p, s)$.

\begin{theorem}\label{thm:rcpdp}
\RCP{} is FPT when parameterized by $(p, \min\{s, d\})$.
\end{theorem}
\begin{proof}
Without loss of generality, we may assume $P = \Res$ as well as $N(u) \neq \emptyset$ for all $u \in \Users$.
For all $C \subseteq P$, let $U_C = \{u \in \Users: N(u) = C\}$ (notice that we might have $U_C = \emptyset$ for some $C \subseteq P$). 
Let $S \subseteq \Users$ be a blocker set of size at most $s$, \ie a set whose removal makes $\res(P, 0, d, t)$ unsatisfiable. Moreover, assume that $S$ is a \emph{minimal blocker set}, meaning that there does not exist $S' \subsetneq S$ such that the removal of $S'$ makes $\res(P, 0, d, t)$ unsatisfiable. 
\begin{claim} \label{claim:blocker}
For all $C \subseteq P$, $U_C \cap S \neq \emptyset$ implies that $|U_C \setminus S| < d$.
\end{claim}
Before proving the claim, notice that for all $u \in U_C \cap S$, there exists a set of teams $V = \{V_1, \dots, V_d\}$ such that \begin{inparaenum}[(i)] \item $\U(V) \cap S = \{u\}$, \label{cond:singleton} and \item $|\U(V) \cap U_C| \le d$\label{cond:intersectbis}\end{inparaenum}. Condition (\ref{cond:singleton}) comes from the minimality of $S$, while Condition (\ref{cond:intersectbis}) comes from the fact that otherwise, there would exist $i \in [d]$ such that $|V_i \cap U_C| \ge 2$, and removing one user from $V_i$, arbitrarily chosen in $(V_i \cap U_C) \setminus \{u\}$, produces another set of teams $V'$ with $\U(V') \subsetneq \U(V)$ (with exactly one element less) and still such that $V \cap S = \{u\}$. 
Applying this strategy iteratively, we can get a set of teams $V$ as desired. 

\begin{proof}[of the claim]
To do so, let $u \in U_C \cap S$ and $V= \{V_1, \dots, V_d\}$ defined as previously. 
If we have $|U_C \setminus S| \ge d$, then there exists $v \in U_C \setminus S$ such that $v \notin \U(V)$ (since $|\U(V) \cap U_C| \le d$, and since $u \in S \cap U_C$, it follows that $|(U_C \setminus S) \cap \U(V)| \le d-1$), in which case we have that $(\U(V) \setminus \{u\}) \cup \{v\}$ is the union of a set of teams which does not intersect $S$ (recall that $\U(V) \cap S = \{u\}$), and satisfies $\res(P, 0, d, t)$ (since $N(u) = N(v)$), a contradiction.\qed
\end{proof}

We now define a reduced set of users $\Users^r \subseteq \Users$ composed of $d_C=\min\{|U_C|, d\}$ users from $U_C$ chosen arbitrarily, for all $C \subseteq P$.
By construction, observe that $|\Users^r| \le d2^p$. We also define, for all $C \subseteq P$, $\Users_C^r = \Users_C \cap \Users^r$. 
Finally, consider an algorithm which outputs that $\res(P,s,d,t)$ is unsatisfiable if and only if there exists a blocker set $S \subseteq \Users^r$ of the instance induced by $\Users^r$ (\ie with authorization policy $\ur|_{\Users^r}$), and such that $\sum_{C \subseteq P} \zeta_S(C) \le s$, where
\[
  \zeta_S(C) = 
   \begin{cases}
    |S \cap \Users_C^r| + |U_C| - d_C & \text{if } S \cap \Users_C^r \neq \emptyset \\
    0 & \text{otherwise}.
   \end{cases}
 \]
in which case we will say that $S$ is a \emph{reduced blocker set}. 
We will prove that this algorithm is FPT parameterized by $(p, \min\{s, d\})$, and is correct.

Concerning the running time, observe first that the construction of $\Users^r$ as well as the evaluation of $\zeta_S$, given $S \subseteq \Users^r$, takes $O^*(2^p)$ time.
Then, for any reduced blocker set $S \subseteq \Users^r$, notice that $|S \cap \Users^r_C| \le \min\{s, d\}$ for all $C \subseteq P$, and that any set $S' \subseteq \Users^r$ such that $|S' \cap \Users^r_C| = |S \cap \Users^r_C|$ for all $C \subseteq P$ is also a reduced blocker set (because $N(u) = N(v)$ for all $u, v \in C$, for all $C \subseteq P$). 
Hence, instead of enumerating every possible subset $S$ of $\Users^r$, it is sufficient to enumerate the sizes of each intersection with $\Users^r_C$ for all $C \subseteq P$, and pick the right number of users in $\Users^r_C$ in an arbitrary way. 
Since its intersection is of size at most $\min\{s, d\}$, the number of sets to enumerate is $O((\min\{s, d\}+1)^{2^p})$.
Then, for each obtained set $S \subseteq \Users^r$, we can check whether it is a blocker set of $\ur|_{\Users^r}$ by solving the \RCP{s=0} problem on the instance $\ur|_{\Users^r \setminus S}$ in FPT time parameterized by $p$ (using, \eg, Theorem~\ref{thm:rcp-parameterized-by-p-only}).

It now remains to prove its correctness, by proving that there exists a reduced blocker set if and only if $\res(P, s, d, t)$ is unsatisfiable.
If such a set $S$ exists, then define, for each $C \subseteq P$, a set $S_C \subseteq U_C$ composed of $S \cap U_C^r$ plus all users in $U_C \setminus U_C^r$. By construction, $|S_C| = \zeta_S(C)$, and thus $S^* = \bigcup_{C \subseteq P} S_C$ contains at most $s$ users. 
We now prove that $S^*$ is a blocker set: suppose by contradiction that there exists a set of teams $V=\{V_1, \dots, V_d\}$ such that $\U(V) \cap S^* = \emptyset$. As we saw previously, we may assume that $|V_i \cap U_C| \le 1$ for all $i \in [d]$ and all $C \subseteq P$. Let $I_V = \{i \in [d]: V_i \cap (U_C \setminus U_C^r) \neq \emptyset\}$. We show that we can turn $V$ into another set of teams $V'$ such that $\U(V') \subseteq \Users^r$ (\ie such that $I_{V'} = \emptyset$), implying that $S$ is not a reduced blocker set, a contradiction. If $I_V = \emptyset$, then we are done. Otherwise let $i \in I_V$ and $u \in V_i \cap (U_C \setminus U_C^r)$. By construction of $U^r$, there exists $v \in U_C^r$, and thus $(V \setminus \{u\}) \cup \{v\}$ is the union of a set of teams $V'$ (recall that $N(u) = N(v)$) such that $i \notin I_{V'}$. Repeating this transformation at most $d$ times, we naturally obtain a set of teams $V'$ such that $I_{V'} = \emptyset$ as desired.

Conversely, suppose that $\res(P, s, d, t)$ is unsatisfiable, \ie there exists a blocker set of users $S \subseteq \Users$ of size at most $s$. As previously, we may assume that $S$ is a minimal blocker set. 
We now use the previous Claim, and thus for all $C \subseteq P$, $|S \cap U_C| \ge \max\{0, |U_C|-d+1\}$. Thus, we may assume, without loss of generality (since, again, $N(u) = N(v)$ for all $u, v \in U_C$) that $U_C \setminus U_C^r \subseteq S$. Then, we define $S^r = S \setminus (\bigcup_{C \in \wp(C)} U_C \setminus U_C^r)$. Observe that for all $C \subseteq P$, we have:
\begin{eqnarray*}
 \zeta_{S^r}(C) & = & |S^r \cap \Users^r_C| + |U_C|-d_C \\
 				& = & |S^r \cap \Users^r_C| + |U_C \setminus U_C^r| \\
 				& = & |S \cap U_C|
\end{eqnarray*}
and thus $\sum_{C \subseteq P}\zeta_{S_r}(C) = \sum_{C \subseteq P} |S \cap U_C| = |S| \le s$. Finally, $S^r$ is indeed a blocker set of the instance induced by $\Users^r$, since otherwise, there would exist a set of teams $V= \{V_1, \dots, V_d\}$ with $\U(V) \subseteq \Users^r$ such that $\U(V) \cap S^r = \emptyset$, which would imply that $\U(V) \cap S = \emptyset$ as well, a contradiction.\qed
\end{proof}

\subsection{Negative results}\label{sec:negres}

It is worth pointing out that the reduction of \cite[Lemma 3]{LiWaTr09} proving the NP-hardness of \RCP{s=0, d=1} actually proves the W[$2$]-hardness of this problem parameterized by $t$ (from \textsc{Set Cover} parameterized by the size of the solution \cite{DoFe13}). Another implication of this reduction is the para-NP-hardness of \RCP{} when parameterized by $(s, d)$. We now complement this result by showing that \RCP{d=1, t=\tau} is coNP-hard for every fixed $\tau \ge 3$, implying para-coNP-hardness of \RCP{} parameterized by $(d, t)$. The result is obtained by a reduction from the \textsc{$\delta$-Hitting Set} problem for every $\delta \ge 2$.

\begin{theorem}\label{thm:paranphard-hittingset}
\RCP{d=1, t=\tau} is coNP-hard for every fixed $\tau \ge 3$.
\end{theorem}
\begin{proof}
We reduce from the \textsc{$\delta$-Hitting Set} problem, in which we are given a ground set $V=\{v_1, \dots, v_n\}$, a set $S = \{S_1, \dots, S_m\}$ with $S_j \subseteq V$ and $|S_j|=\delta$ for all $j \in [m]$ and an integer $k$, and where the goal is to find a set $C \subseteq V$ of size at most $k$ and such that $C \cap S_j \neq \emptyset$ for all $j \in [m]$. This problem is known to be NP-hard for every $\delta \ge 2$ \cite{GaJo79}.

Hence, let $(V, S, k)$ be an instance of \textsc{$\delta$-Hitting Set} defined as above. For every $j \in [m]$, fix an arbitrary ordering of $S_j$, which can thus be seen as a tuple $(v_{i_1}, \dots, v_{i_{\delta}})$, allowing us to define $S_j[x] = v_{i_x}$ for all $x \in [\delta]$. 

We define a set of users $\Users = \Users^V \cup \Users^S$, where $\Users^V = \{u^V_1, \dots, u^V_n\}$ and $\Users^S = \{u^S_1, \dots, u^S_m\}$. We then define a set of resources $\Res = \Res^V \cup \Res^S \cup \{r^*\}$, where $\Res^S = \bigcup_{j = 1}^m P^j$ with $P^j = \{p^j_1, \dots, p^j_{\delta}\}$ for all $j \in [m]$, and where $\Res^V$ contains one resource $r^V_{Q}$ for every subset $Q$ of $\delta-1$ users of $\Users^V$.\\
We now define the authorization policy $\ur$ by giving $N(u)$ for every $u \in \Users$. For every $i \in [n]$, $N(u^V_i)$ is composed of $\{p^j_x: j \in [m], x \in [\delta]$ such that $S_j[x] = v_i\}$ together with all resources $r^V_Q$ such that $u^V_i \notin Q$, for every subset $Q$ of $\delta-1$ users of $\Users^V$. For all $j \in [m]$, $N(u^S_j)$ is composed of $r^*$ together with $\Res^S \setminus P^j$. To conclude the construction, we let $P = \Res$, $t = \delta +1$, $d=1$, and $s=k$. Clearly this reduction can be done in polynomial time.

The remainder consists in proving that every team (\ie sets of at most $t$ users having collectively access to all $\Res$) is of the form $T_j=\{u^V_{i_1}, \dots, u^V_{i_{\delta}}, u^S_j\}$ such that $S_j = \{v_{i_1}, \dots, v_{i_{\delta}}\}$. 
If this is true, then observe that since, for every $j \in [m]$, user $u^S_j$ only belongs to team $T_j$, we will be able to suppose \emph{w.l.o.g.} that it does not belong to any blocker set, and thus the set of all teams will be in one-to-one correspondance with the sets in $S$, implying that the obtained instance has a blocker set of size at most $s$ ($=k$) if and only if there is a hitting set of size at most $k$.

Let $T \subseteq \Users$ of size at most $t$. By construction, we need at least $\delta$ users from $\Users^V$ to have access to all resources in $\Res^V$ (indeed, every set $Q$ of $\delta-1$ users from $\Users^V$ has only access to $\Res^V \setminus \{r^V_Q\}$), and we also need at least one user from $\Users^S$ to have access to $r^*$. Hence, $|T \cap \Users^V| = \delta$ and $T \cap \Users^S = \{u^S_j\}$ for some $j \in [m]$. Now, notice that $u^S_j$ has access to all resources in $\Res$ but $P^j$, which implies that $T \cap \Users^V$ must have collectivelly access to all resources in $P^j$. However, this can only happen if $T \cap \Users^V = \{u^V_{i_1}, \dots, u^V_{i_{\delta}}\}$, where $S_j = \{v_{i_1}, \dots, v_{i_{\delta}}\}$, concluding the proof.\qed
\end{proof}

We also settle the case of \RCP{} parameterized by $(s, t)$ (and thus \RCP{s=0} parameterized by $t$). The result is obtained by a reduction from the \TDM problem.

\begin{theorem}\label{thm:paranphard-matching}
\RCP{s=0, t=4} is NP-hard.
\end{theorem}
\begin{proof}[of Theorem~\ref{thm:paranphard-matching}]
We reduce from the \TDM problem, in which we are given three sets $X$, $Y$ and $Z$ of $n$ elements each, a set $M \subseteq X \times Y \times Z$ of hyperedges, and an integer $k$. The goal is to find $M' \subseteq M$ with $|M'| \ge k$ such that $\forall e, e' \in M'$ with $e \neq e'$, $e=(x, y, z)$, $e'=(x', y', z')$, we have $x \neq x'$, $y \neq y'$ and $z \neq z'$ (in that case, we will say that these two hyperedges are \emph{disjoint}). We note $m=|M|$, $X = \{x_1, \dots, x_n\}$, $Y = \{y_1, \dots, y_n\}$ $Z = \{z_1, \dots, z_n\}$, and $M = \{e_1, \dots, e_m\}$.

We then define the following set of resources:
\begin{eqnarray*}
 \P & = & \{r_1^X, \dots, r_m^X\} \\
 & & \cup \{r_1^Y, \dots, r_m^Y\} \\
 & & \cup \{r_1^Z, \dots, r_m^Z\} \\
 & & \cup \{r_X, r_Y, r_Z, r_*\}
 \end{eqnarray*}
  and a set of users $\Users$ composed of $\Users_X$, $\Users_Y$, $\Users_Z$ and $\Users^*$, where, for all $\omega \in \{X, Y, Z, *\}$, we note $\Users_{\omega} = \{u^{\omega}_1, \dots, u^{\omega}_n\}$. Then the authorization policy $A$ is informally constructed as follows: for each hyperedge $e_j = \{x_{i_1}, y_{i_2}, z_{i_3}\}$, user $u^X_{i_1}$ (resp. $u^Y_{i_2}$, $u^Z_{i_3}$) has access to $r_j^X$ (resp. $r_j^Y$, $r_j^Z$) and to $r_X$ (resp. $r_Y$, $r_Z$), and user $u^*_j$ has access to all resources in $\P$ but $r_j^X$, $r_j^Y$, $r_j^Z$, $r_X$, $r_Y$ and $r_Z$. More formally, we have:
\begin{eqnarray*}
A & = & \{(u^X_i, r_j^X) : x_i \text{ belongs to } e_j, \forall i \in [n], \forall j \in [m]\} \\
	&& \cup \{(u^Y_i, r_j^Y) : y_i \text{ belongs to } e_j, \forall i \in [n], \forall j \in [m]\} \\
	&& \cup \{(u^Z_i, r_j^Z) : z_i \text{ belongs to } e_j, \forall i \in [n], \forall j \in [m]\} \\
	&& \cup \{(u^*_j, r_{h}^{\omega}) : \forall \omega \in \{X, Y, Z\}, \forall j, h \in [m] j \neq h\} \\
	&& \cup \{(u^{\omega}_i, r_{\omega}) : \forall i \in [n], \forall \omega \in \{X, Y, Z\}\} \\
	&& \cup \{u^*_j, r_*) : \forall j \in [m]\} 
\end{eqnarray*}
To conclude the construction, which can be done in polynomial time, we set $d = k$, and the resiliency policy is thus $\res(\P, 0, d, 4)$.

First, suppose that there exists a solution $M'$ for the \TDM problem. Without loss of generality, assume that $|M'| = k$, $M' = \{e_1, \dots, e_k\}$, and that $e_i = (x_i, y_i, z_i)$ for all $i \in [k]$ (recall that all members of $M'$ are pairwise disjoint). Then, observe that for all $i \in [k]$, user $u_i^*$ has access to all resources but $r_i^X$, $r_i^Y$, $r_i^Z$, $r_X$, $r_Y$ and $r_Z$. However, $u^X_i$ has access to $r_i^X$ and $r_X$, user $u^Y_i$ has access to $r_i^Y$ and $r_Y$, and user $u^Z_i$ has access to $r_i^Z$ and $r_Z$. Hence, we have $N(\{u^X_i, u^Y_i, u^Z_i, u^*_i\}) = \P$, and, since all members of $M'$ are pairwise disjoint, we thus constructed a set of teams for \RCP{s=0, t=4}, or, in other words, $\res(\P, 0, d, 4)$ is satisfiable.

Conversely, suppose that there exist $V_1, \dots, V_d$, pairwise disjoint subsets of $\Users$ such that for all $i \in [d]$, we have $|V_i| = 4$ and $N(V_i) = P$. We first claim that for all $i \in [d]$, $V_i$ intersects $U_X$ (resp. $U_Y$, $U_Z$ and $U_*$) on exactly one element.
Indeed, otherwise, since $|V_i| = 4$ and since all users in $U_X$ (resp. $U_Y$, $U_Z$, $U_*$) have access to only $r_X$ (resp. $r_Y$, $r_Z$, $r_*$) among $\{r_X, r_Y, r_Z, r_*\}$, $V_i$ could not have access to all these resources. Thus, we know that for all $i \in [d]$, we have $V_i = \{u^X_{i_1}, u^Y_{i_2}, u^Z_{i_3}, u^*_{i_4}\}$, for some $(i_1, i_2, i_3, i_4) \in [n] \times [n] \times [n] \times [m]$. We claim that $(x_{i_1}, y_{i_2}, z_{i_3}) = e_{i_4}$. Indeed, observe that user $u^*_{i_4}$ has access to all resources but $r_{i_4}^X$, $r_{i_4}^Y$, $r_{i_4}^Z$, $r_X$, $r_Y$ and $r_Z$. By construction, the only way for having $N(V_i) = P$ is that user $u_{i_1}^X$ (resp. $u_{i_2}^Y$, $u_{i_3}^Z$) has access to resources $r_{i_4}^X$ (resp. $r_{i_4}^Y$, $r_{i_4}^Z$) or, in other words, that $x_{i_1}$ (resp. $y_{i_2}$, $z_{i_3}$) belongs to hyperedge $e_{i_4}$. Thus, there exists $k$ pairwise disjoint hyperedges in $M$.
\qed
\end{proof}

\section{Refined positive results for the case $s=0$}\label{sec:sZero}

We now turn to the particular case where $s=0$. As said in Section~\ref{sec:intro}, one motivation for studying this case is that it is the bottleneck of the algorithm of Li \etal~\cite{LiWaTr09} for \RCP{}. 
Hence, we believe that designing efficient algorithms for this sub-case might help us solve much larger instances of \RCP{} than is currently possible. 
To this end, we now provide a complete characterization of the complexity when considering all possible combinations of parameters among $p$, $d$ and $t$. 
We also investigate the question of reduction rules within the framework of kernelization, highlighting a difference of behavior between \RCP{s=0} and \RCP{s=0, t=\infty}.

\subsection{FPT algorithms}

The first algorithm is a dynamic programming-based approach similar to the one for \textsc{Set Cover}~\cite{DoFe13}, in order to obtain an FPT algorithm for \RCP{s=0} parameterized by $(p, d)$. While this result was already known, given that \RCP{} is itself FPT with this parameterization (and that \RCP{s=0} is actually FPT parameterized by $p$ only, as we will see in Theorem~\ref{thm:rcp-parameterized-by-p-only}), we provide for \RCP{s=0} a better running time. In particular, as we will see later, a previous known reduction of Li \etal~\cite{LiWaTr09} actually proves that when $d$ is fixed, the obtained running time is the best we can hope for, under the Exponential Time Hypothesis (ETH)\footnote{The ETH claims that \SAT cannot be solved in $O^*(2^{o(n)})$, where $n$ is the number of variables in the CNF formula~\cite{ImPaZa01}.}.

\begin{theorem}\label{thm:DP}
\RCP{s=0} can be solved in $O^*(2^{d\p})$ time.
\end{theorem}
\begin{proof}[of Theorem~\ref{thm:DP}]
Let $\Users = \{u_1, \dots, u_n\}$.
We define a dynamic programming algorithm which, given any $i \in [n]$ and any $d$-tuple of subsets of $\P$ $(S_1, \dots, S_d)$, returns $yes$ if there exist $d$ mutually disjoint sets $T_1, \dots, T_d$, each being a subset of $\{u_1, \dots, u_i\}$ and such that $S_j \subseteq N(T_j)$ for all $j \in [d]$, and returns $no$ otherwise (in which case we will say that such an algorithm is \emph{correct}). To do so, we define the following recursive formula $DP$. First, we set:
\begin{align*}
DP(0, S_1, \dots, S_d) &= 1 \text{ if and only if } S_j = \emptyset, \text{ for all } j \in [d]; \\
DP(i, \emptyset, \dots, \emptyset) &= 1 \text{ for all } i \in [n].
\end{align*}
For the induction, let $i \in [n]$ and $\S = (S_1, \dots, S_d)$ where $S_j \subseteq \P$ for all $j \in [d]$. 
Let $J=\{j \in [d] : S_j \neq \emptyset\}$, and for all $j \in J$, define $\S_j = (S_1, \dots, S_{j-1}, S_j \setminus N(u_i), S_{j+1}, \dots, S_d)$. Finally, we set:
$$DP(i, \S) = DP(i-1, \S) \vee \left( \bigvee_{j \in J} DP(i-1, \S_j) \right)$$

\begin{lemma}
$DP(i, \S)$ is correct.
\end{lemma}
\begin{proof}
Suppose that $T_1, \dots, T_d$ are $d$ mutually disjoint subsets of $\{u_1, \dots, u_i\}$ such that $S_j \subseteq N(T_j)$ for all $j \in [d]$. We may assume that $T_j \neq \emptyset$ iff $S_j \neq \emptyset$. Then, either $u_i \notin T_j$ for all $j \in [d]$, in which case $DP(i-1, \S)$ returns $yes$, or $u_i \in T_j$ for some $j \in [d]$, which implies $S_j \neq \emptyset$ and thus $j \in J$. In this case $DP(i-1, \S_j)$ returns $yes$.

Conversely, if $DP(i-1, \S)$ return $yes$, then there exist $d$ mutually disjoint sets $T_1, \dots, T_d$, each being a subset of $\{u_1, \dots, u_{i-1}\}$ (and thus a subset of $\{u_1, \dots, u_{i}\}$), and such that $S_j \subseteq N(T_j)$ for all $j \in [d]$. If $DP(i-1, \S_j)$ returns $yes$ for some $j \in J$, then there exist $d$ mutually disjoint sets $T_1, \dots, T_d$, each being a subset of $\{u_1, \dots, u_i\}$ and such that $S_q \subseteq N(T_q)$ for all $q \in [d]$, $q \neq j$, and $S_j \setminus N(u_i) \subseteq N(T_j)$. In this case $S_j \subseteq N(T_j \cup \{u_i\})$.
\qed
\end{proof}
Clearly, $DP(n, \P, \dots, \P)$ returns $yes$ if and only if $res(P, 0, d, t)$ is satisfiable.
A table of size $n2^{d\p}$ is sufficient to store all intermediate results, while each step takes $O(d)$ time, establishing the claimed running time.
\qed
\end{proof}

Li \etal \cite{LiWaTr09} showed that \RCP{s=0,t=\infty, d=3} is NP-hard, by a reduction from \textsc{$3$-Domatic Partition}, which transforms a graph of $n$ vertices into an instance $(\Users, \Res, \ur, \res(\P, 0, 3, \infty))$ with $|\P|=n$. Since a $2^{o(n)}$ algorithm for \textsc{$3$-Domatic Partition} would violate the ETH (by a linear reduction from \textsc{SAT} \cite{Cre95}), we have the following:
\begin{theorem}\label{thm:lowerbounds}
\RCP{s=0, t=\infty, d=3} cannot be solved in $2^{o(p)}$ time unless the ETH fails.
\end{theorem}

Hence, for fixed $d$, the algorithm described in Theorem~\ref{thm:DP} has an optimal running time.
We continue our quest for a better understanding of the frontier between tractable and intractable cases of the \RCP{s=0} problem. 
Given the positive result parameterized by $(p, d)$, a natural question is to consider each parameter separately. 
The question can well be answered negatively concerning the parameter $d$, since, as we saw before, \RCP{s=0, d=3, t=\infty} is NP-hard \cite{LiWaTr09}, and thus \RCP{s=0} is para-NP-hard parameterized by $d$. 
However, we are able to give a different answer for the parameter $\p$ only.

\begin{theorem}\label{thm:rcp-parameterized-by-p-only}
\RCP{s=0} is FPT when parameterized by $\p$.
\end{theorem}
\begin{proof}
The result makes use of Lenstra's celebrated algorithm~\cite{Len83} for Integer Linear Programming Feasibility (\ILPF) parameterized by the number of variables. 

\begin{theorem}[Lenstra~\cite{Len83}]
Whether a given ILP has a non-empty solution set can be decided in $O^*(f(n))$ time for some computable function $f$, where $n$ denotes the number of variables of the ILP.
\end{theorem}

Note that this algorithm has been improved by Kannan~\cite{Ka87}, with $f(n) = n^{O(n)}$ (but exponential space), and by Frank and Tardos~\cite{FrTa87} so that the algorithm runs in polynomial space, and with $f(n) = O(n^{2.5n + o(n)})$.

We thus give an \ILPF formulation of the problem with a number of variables depending on $\p$ and $t$. As we saw previously, since we may assume that $t \le p$ in any positive instance, the result will follow (by Lenstra's result) for the parameterization by $p$ only.

Let $(\Users, \Res, \ur, \res(P, 0, d, t))$ be the input instance of \RCP{s=0}. For any $N \subseteq \P$, let $U_N$ denote the set of users having neighborhood exactly $N$ in $\P$, or, formally: $U_N = \{u \in \Users : N(u) = N\}$.
Moreover, we define the following set called \emph{configurations}: 
\[
\C = \left\{\{N_1, \dots, N_b\} : b \le t, N_i \subseteq \P, i \in [b], \bigcup_{i=1}^b N_i = \P\right\}.  
\]

For any $N \subseteq \P$, we note \[ \C_N = \left\{ c = \{N_1, \dots, N_{b_c}\} \in \C : N=N_i \text{ for some } i \in [b_c]\right\}\] the set of configurations involving $N$. Informally, a configuration $\{N_1, \dots, N_b\}$ represents a way to dominate $P$, by picking one user in $U_{N_i}$, for each $i \in [b]$.

The variables of our ILP are in one-to-one correspondence with elements of $\C$, and will be denoted by $\{x_c : c \in \C\}$. Since $\C$ is of size bounded by $O(\sum_{b=1}^t 2^{b\p} )$, the number of variables is bounded by a function of $\p$ and $t$ only. Then, we define the following two sets of constraints:
\begin{enumerate}
	\item \label{const:sumd} $\sum_{c \in \C} x_c = d$,
	\item \label{const:sumUn} $\sum_{c \in C_N} x_c \le |U_N|$ for all $N \subseteq \P$.
\end{enumerate}
We now explain the idea of the ILP. 
Observe that in a positive instance, there always exists a set of teams in which in each set, each user has a different neighborhood. For any $T \subseteq \Users$, define $\phi(T) = \{N(u): u \in T\}$, the set of neighborhoods of users in $T$. Then, by definition of the problem, for any set of teams $V = \{T_1, \dots, T_d\}$, we have $\Phi(T_i) \in \C$ for all $i \in [d]$. Notice that we might have $\Phi(T_i) = \Phi(T_j)$ for $i, j \in [d]$, $i \neq j$. We can associate, with each such set of teams, a vector $X^V = \{x^V_c\}_{c \in \C}$, where $x^V_c$ is the number of sets of $V$ having configuration $c \in C$. By the remark above, we might have $X^V = X^{V'}$ for two different sets of teams $V$ and $V'$, in which case we will say that these two sets of teams are \emph{configuration-equivalent}. Observe that given a vector $X=\{x_c\}_{c \in \C}$ such that $X=X^{V^*}$ for a fixed set of teams $V^*$, we can construct in polynomial time a set of teams $V$ that is configuration-equivalent to $V^*$; constraints~(\ref{const:sumd}) and~(\ref{const:sumUn}) aim to find such a vector.
Suppose that there exists a set of teams $V^* = \{T_1, \dots, T_d\}$ of the problem. It is clear that $X^{V^*}$ fulfills constraints~(\ref{const:sumd}) and~(\ref{const:sumUn}). Conversely, constraints in~(\ref{const:sumd}) ensure that the set of teams will contain $d$ sets, while constraints in~(\ref{const:sumUn}) ensure that when constructing a set of configuration $c = \{N_1, \dots, N_{b_c}\}$, there must exist a new user having neighborhood exactly $N_i$ for all $i \in [b_c]$ and that has not been already assigned to another set. 
\qed
\end{proof}

\subsection{User reductions}\label{sec:kernel}

We now focus on reduction rules which can be performed in polynomial time and result in an equivalent instance having a smaller number of users. More formally, we say that a (decision) problem has a \emph{kernel}~\cite{DoFe13} of size $f$, for some computable function $f:\mathbb{N} \rightarrow \mathbb{N}$, if there exists a polynomial algorithm which, given an instance $x$ with parameter $k$, outputs an instance $x'$ of size $|x'|$ with parameter $k'$ such that:
\begin{inparaenum}[(i)]
	\item $k' \le k$, 
	\item $x$ is positive if and only if $x'$ is positive, and
	\item $|x'| \le f(k)$.
\end{inparaenum}
In the case of \RCP{s=0} our aim is thus to obtain an equivalent instance with a number of users bounded by a function of $d$ and $t$.

While the role of $t$ was so far of less interest for the complexity of the problem, we show that the problem behaves differently from the kernelization point of view, depending on whether $t=\infty$ or not. We first show that when $t=\infty$, the problem admits a kernel with at most $dp$ users.
To do so, we will make use of the following:

\begin{lemma}[$d$-expansion Lemma~\cite{CyFoKoLoMaPiPiSa15}]\label{lemma:dexpansion}
Let $d \ge 1$ be a positive integer and $G=(A, B, E)$ be a bipartite graph with bipartition $(A, B)$ and $E \subseteq A \times B$ such that for all $b \in B$, $N(b) \neq \emptyset$. If $|B| \ge d|A|$, then there exist non-empty vertex sets $X \subseteq A$ and $Y \subseteq B$ which can be found in time polynomial in the size of $G$, such that:
\begin{enumerate}[(i)]
	\item $N(Y) \subseteq X$, and
	\item there is a $d$-expansion of $X$ into $Y$: a collection $M \subseteq E \cap (X \times Y)$ such that every vertex of $X$ is incident to exactly $d$ edges of $M$, and exactly $d|X|$ vertices of $Y$ are incident to an edge of $M$.
\end{enumerate}
\end{lemma}

\begin{theorem}\label{thm:kernel}
\RCP{s=0, t=\infty} admits a kernel with at most $dp$ users.
\end{theorem}
\begin{proof}
Suppose we are given an instance of \RCP{s=0, t=\infty}. We present two reduction rules which are used to decrease the number of users. For each of these rules, we will prove that the instance is positive iff the reduced instance is positive, in which case we will say that the rule is \emph{safe}.\\

\textit{Reduction Rule 1:} if there exists $u \in \Users$ with $N(u) = \emptyset$, then delete $u$.
\begin{proof}[of safeness]
Simply observe that such a user cannot participate in any set of teams if the instance is positive, and, conversely, cannot turn a negative instance into a positive one if it is deleted. \qed
\end{proof}

\textit{Reduction Rule 2:} if there exist $X \subseteq \P$, $Y \subseteq \Users$ such that $N(Y) \subseteq X$ and there is a $d$-expansion of $X$ into $Y$, then delete $X$ from $\P$, $Y$ from $\Users$, and $(Y \times X) \cap \ur$ from $\ur$.
\begin{proof}[of safeness]
If the instance is a positive one, then there exists a set of teams $\{V_1, \dots, V_d\}$. Then, for all $r \in \P \setminus X$, there does not exist $u \in Y$ such that $(u, r) \in \ur$, since $N(Y) \subseteq X$. Hence, $N(V_i \setminus Y) \supseteq \P \setminus X$, and thus $\{V_1 \setminus Y, \dots, V_d \setminus Y\}$ is a set of teams for the reduced instance, which is thus a positive one. 

Conversely, suppose that the reduced instance is a positive one: there exist $V_1, \dots, V_d$, disjoints sets of users from $\Users \setminus Y$ such that $N(V_i) \supseteq \P \setminus X$. Since there is a $d$-expansion of $X$ into $Y$, for all $r \in X$, there exist $u^r_1, \dots u^r_d \in Y$ such that $(u^r_i, r) \in \ur$ for all $i \in [d]$, where $u^r_i \neq u^{r'}_{i'}$ for all $r \neq r'$ and $i \neq i'$. Hence, for all $i \in [d]$, if we set $V_i' = V_i \cup \{u^r_i : r \in X\}$, we have $V_i' \cap V_j' = \emptyset$ for all $1 \le i < j \le d$, and $N(V_i') \supseteq \P$ for all $i \in [d]$, and thus we have a positive instance as well, which proves that the rule is safe. \qed
\end{proof}

Since each reduction rule can be applied in polynomial time, and since each of them decreases the number of users by at least one, the algorithm runs in polynomial time. Finally, by Lemma~\ref{lemma:dexpansion}, if none of the previous reduction rules applies, then $|\Users| \le dp$, and we thus have a kernel with at most $dp$ users, as desired. \qed
\end{proof}

As Li \etal \cite{LiWaTr09} point out, \RCP{s=0, d=1} is equivalent to the \textsc{Set Cover Problem}. Known kernel lower bounds for this problem \cite{DoLoSa09} lead to the following theorem, which is in sharp contrast to the previous case.

\begin{theorem}\label{thm:kernelLowerBound}
\RCP{s=0, d=1} (and thus \mbox{\RCP{s=0}}) does not admit a kernel with $(\p+t)^{O(1)}$ users, unless $\text{coNP} \subseteq \text{NP/poly}$.
\end{theorem}

\section{Conclusion and future work}\label{sec:conclusion}
We considered \RCP{}, a problem introduced recently in the area of access control to analyze the resiliency of a system. 
Given the large number of natural parameters in an instance of this problem, and given that these parameters are likely to take small values in practice, our goal was to provide a systematic analysis of the complexity of the problem using the framework of parameterized complexity. 
For all but one possible combination of the parameters, we were able to obtain either a positive or negative result. We also considered a restricted variant of the problem for which we settled the parameterized complexity of all possible combinations of the parameters.
A first obvious idea of future work is thus to fill the remaining hole of Figure \ref{fig:lattices}, namely to decide whether \RCP{} is in FPT, XP, W[1]-hard or para-(co)NP-hard parameterized by $p$.

Another interesting further line of research would be to study resiliency aspects with respect to other problems. In the context of graphs for instance, we could define the problem of determining whether upon removal of at most $s$ vertices, a given graph still satisfies some property given by another combinatorial problem, \eg having a vertex cover of size $k$. We believe that considering structural parameterizations (together with $s$) might lead to interesting new results. As in our case, the complexity of such a new problem will certainly depend on the complexity of the considered underlying problem (\ie the case $s=0$).
%

\bibliographystyle{plain}
\bibliography{refs}

%

%
%

%
%

\end{document}